\documentclass[12pt, a4paper]{article}
            
\usepackage{a4}
\usepackage{graphicx}
\usepackage{amsmath}
\usepackage{amssymb}
\usepackage{mathrsfs}      
\usepackage{amsthm}
\usepackage{enumerate}

\usepackage{hyperref}
\hypersetup{pdftitle={Cohomology of Line Bundles: Proof of the Algorithm}, 
      pdfauthor={Helmut Roschy, Thorsten Rahn}, 
      pdfsubject={Algebraic Geometry, Combinatorial Algebra, Mathematical Methods in Physics},
      pdfstartview={FitH}, pdfpagelayout={TwoColumnRight},
      bookmarksopen, bookmarksnumbered, bookmarksopenlevel=2}

\newcommand{\beq}{\begin{equation}}  \newcommand{\eeq}{\end{equation}}
\newcommand{\bal}{\begin{aligned}}   \newcommand{\eal}{\end{aligned}}

\newtheorem*{theorem}{Theorem}
\newtheorem*{lemma}{Lemma}

\def\Cl{\text{Cl}} 

\def\IC{\mathbb{C}}

\def\IZ{\mathbb{Z}}
\def\IN{\mathbb{N}}

\def\cO{\mathcal{O}}
\def\cS{\mathcal{S}}
\def\cP{\mathcal{P}}

\def\cF{\mathcal{F}}
\def\cT{\mathcal{T}}

\def\fm{\mathfrak{m}} 

\def\bx{\mathbf{x}}
\def\be{\mathbf{e}}
\def\bu{\mathbf{u}}
\def\ba{\mathbf{a}}
\def\bb{\mathbf{b}}

\def\wH{\widetilde{H}}
\def\wC{\widetilde{\mathcal{C}}}
\def\wsig{\sigma}
\def\larr{\longleftarrow}
\def\rarr{\longrightarrow}
\def\ds{\partial}

\DeclareMathOperator{\Class}{\mathrm{Cl}}

\DeclareMathOperator{\negind}{\mathrm{neg}}
\DeclareMathOperator{\Tor}{\mathrm{Tor}}
\DeclareMathOperator{\link}{\mathrm{link}}
\DeclareMathOperator{\kernel}{\mathrm{ker}}
\DeclareMathOperator{\im}{\mathrm{im}}
\DeclareMathOperator{\sign}{\mathrm{sign}}

\DeclareMathOperator{\lcm}{\mathrm{lcm}}
\DeclareMathOperator{\degree}{\mathrm{deg}}

\def\clap#1{\hbox to 0pt{\hss#1\hss}}

\def\ce{\mathrel{\mathop:}=}  

\begin{document}

\baselineskip=14pt

\vspace*{-1.5cm}
\begin{flushright}    
  {\small
  MPP-2010-64  }
\end{flushright}

\vspace{2cm}
\begin{center}        
  {\LARGE
  Cohomology of Line Bundles:\\[0.2cm]
  Proof of the Algorithm
  }
\end{center}

\vspace{0.75cm}
\begin{center}        
  Thorsten Rahn$^1$, 
  Helmut Roschy$^1$
\end{center}

\vspace{0.15cm}
\begin{center}        
  \emph{$^{1}$ Max-Planck-Institut f\"ur Physik, F\"ohringer Ring 6, \\ 
               80805 M\"unchen, Germany}
               \\[0.15cm] 
  E-mail: \tt rahn@mppmu.mpg.de, roschy@mppmu.mpg.de
\end{center} 

\vspace{1.5cm}


\begin{abstract}  We present a proof of the algorithm for computing line bundle valued
  cohomology classes over toric varieties conjectured by R.~Blumenhagen, B.~Jurke and the
  authors (\href{http://arxiv.org/abs/1003.5217}{\tt arXiv:1003.5217}) and suggest a kind
  of Serre duality for combinatorial Betti numbers that we observed when computing
  examples.
\end{abstract}


\tableofcontents

 \section{Introduction}

 In a recent publication~\cite{Cohompaper} we conjectured an algorithm\footnote{The
   speed-optimized implementation in C++ can be downloaded from
   \href{http://wwwth.mppmu.mpg.de/members/blumenha/cohomcalg}{http://wwwth.mppmu.mpg.de/members/blumenha/cohomcalg}
   and is regularly updated. To get a first experience of the calculations possible, one
   can also have a quick start with a short Mathematica script that is also available
   there. } facilitating the computation of sheaf cohomology of line bundles over toric
 ambient spaces, quantities that are often required by string theorists since these data
 encode much of the phenomenological content of a chosen string compactification on an
 embedded Calabi-Yau manifold. The arising new computational possibilities in heterotic
 model building or F-theory will be investigated in a separate paper on
 applications~\cite{PaperSecondPart}, but before this we would like to give the algorithm
 a firm mathematical basis in order to dispel any lurking doubts about therewith achieved
 results in advance.

 Before we blithely plunge into the abstract world of combinatorial algebra, we would
 first like to emphasize the advantages that showed up when we first made use of our
 conjectured algorithm. First of all, using Serre duality as a consistency check, we could
 often get rid of explicit maps of complexes -- usually, one has to determine the ranks of
 sparse matrices whose dimensions may be huge for high-dimensional varieties. This sped up
 computations by an estimated (and palpable) factor of $10^{6}$ in comparison with
 existing routines that are used e.g. in Macaulay2~\cite{Macaulay2}.\footnote{In order to
   do sheaf cohomology computations on general toric varieties, the additional package
   ``NormalToricVarieties.m2'' written by Gregory Smith is needed. Since this is still work
   in progress, it is not yet included in the official distribution, but the package
   content can be copied from his
   \href{http://mast.queensu.ca/\~ggsmith/NormalToricVarieties.m2}{homepage}, and then
   seperately loaded into Macaulay2.
 } For most of the examples of toric ambient spaces common to F-theory or other parts of
 theoretical physics computations are instantaneous (at least if one restricts oneselve to
 divisor charges of a sensible size), whereas on the same platform Macaulay2 needs some
 seconds already in the case of 2-dimensional varieties like the del-Pezzo surfaces.

 Another important advantage is the restriction to contributions coming from the power set
 of the Stanley-Reisner ideal of the toric variety. Hence, our algorithm also outstrips
 the ``chamber algorithm'' that is presented in~\cite{CoxLittleSchenk:ToricVarieties} and implemented
 in~\cite{Cvetic:2010rq}. As it turns out, in many examples only $\sim 10 \% $ of all
 possible denominators of representatives for \v{C}ech cohomology appear in the power set,
 so about $90 \%$ of the complexes (with explicit matrices) that are evaluated in the
 chamber algorithm will not contribute from the very beginning, which also accounts for
 the lower performance of this method. Nevertheless, the chamber algorithm served as a
 motivation for our conjecture, since it conveys some intuitive feeling of how local
 sections fit together and how \v{C}ech cohomology arises.\footnote{Note that it can be
   shown that \v{C}ech cohomology on an open cover of a toric variety can be shown to be
   isomorphic to sheaf cohomology, see Theorem 9.0.4 in~\cite{CoxLittleSchenk:ToricVarieties}.}

 The aim of this paper is to substantiate the mathematical background of the conjectured
 algorithm and then give a rigorous proof in the language of combinatorial commutative
 algebra. In fact, when we set out for a mathematical understanding, we found that many of
 the appearing structures could be given a quick formulation in algebraic terms, and some
 features of the algorithm can be derived quite easily in this manner. Therefore, Section
 2 will be devoted to the introduction of some concepts from commutative algebra and the
 beauty and concreteness of these will then enable us to reformulate the original
 conjecture and prove it in Section 3.

 In writing the last sentences of this work, another proof of our conjecture was published
 by S.Y.~Jow in~\cite{Jow}. It puts more emphasis on elements of simplicial topology, so
 we hope that the combinatorial counterpart that we present is of interest on its own. For
 completeness, we will make the equivalence of both approaches explicit right after our
 main theorem. In fact, the key to the meaning of our ``remnant cohomology'' was also
 anticipated by Jow.

 We finish this work by making some comments on further possible simplifications that may
 speed up computations by another of magnitude and in particular point out the possibility
 of ``Serre duality for Betti numbers'' emerging from the established relationships
 between toric and combinatorial algebra.

 \section[Mathematical Preliminaries]{Mathematical Preliminaries}

\subsection{Normal Toric Varieties and Alexander duality}
\label{sec:norm-toric-vari}

Since an introduction to toric varieties was already given in~\cite{Cohompaper}, we
confine ourselves to a short review of the setup. Let $X$ be a complete simplicial smooth
normal toric variety\footnote{in the sense of Chapter 3 of~\cite{CoxLittleSchenk:ToricVarieties}} that
is given in terms of $n$ \textit{vertices} $\nu_{i}$ in $\IZ^{d}$, a
\textit{triangulation} of this vertex scheme yielding the fan $\Sigma$ of the variety $X$
and \textit{weights} (often called charges by the physicists) coming from relations among
the vertices. The homogeneous coordinate ring or \textit{Cox ring} belonging to $X$ is
$S=\IC[\bx]\ce\IC[x_{1},x_{2},\dots ,x_{n}]$ and the \textit{irrelevant ideal}
$B_{\Sigma}$ in $S$ is generated by the monomials
 \begin{equation}
   \label{eq:1}
   \lbrace x_{j_{1}}\cdots x_{j_{s}}\, |\,\lbrace \nu_{1},\dots ,\nu_{n}\rbrace \setminus
   \lbrace \nu_{j_{1}},\dots ,\nu_{j_{s}}\rbrace \,\, \text{spans a cone in } \Sigma
   \rbrace \, ,
 \end{equation}
 so a minimal generating set of $B_{\Sigma}$ is given by the monomials corresponding to
 complements of the maximal cones of $\Sigma$. 
 The \textit{Stanley-Reisner ideal} $I_{\Sigma}$ in $S$ is generated by the monomials
 \begin{equation}
   \label{eq:2}
   \lbrace x_{i_{1}}\cdots x_{i_{s}}\, |\,
   \nu_{i_{1}},\dots ,\nu_{i_{s}}\,\, \text{do not lie in a common cone of } \Sigma
   \rbrace \, ,
 \end{equation}
 so a minimal set of generators would be given by monomials corresponding to minimal
 sets of vertices that do \textit{not} span a cone in $\Sigma$. 

 We want to point out that there is a nice duality operation for monomial ideals that
 connects $B_{\Sigma}$ and $I_{\Sigma}$. Since it will be a key ingredient to account for
 the appearance of $I_{\Sigma}$ in the computation of sheaf cohomology, we introduce it
 more formally. First note that the fans we discern are all simplicial, i.e.~for all cones
 in $\Sigma$ the generating vertices are linearly independent vectors in $\IZ^{d}$. This
 means that instead of the fan $\Sigma$, we may also work with the corresponding
 simplicial complex $\Delta$ of the variety defined on the set $[n]=\lbrace 1,\dots ,
 n\rbrace$, where the $i$-dimensional faces of $\Delta$ are in one-to-one correspondence
 with the $(i+1)$-dimensional cones of $\Sigma$. In the remainder of the paper, we will
 sometimes assume this perspective of things and restrict ourselves to the language of
 simplicial complexes.\footnote{A condensed introduction to simplicial complexes meeting
   our requirements is given e.g.~by the first chapter of~\cite{Sturmfels}.} That means
 that we also write $I_{\Delta}$, $B_{\Delta}$, $\dots$, but one should keep in mind that
 any statements in this language can in principle be given a topological meaning when
 talking about the geometry of polyhedral cones rather than the algebra of simplicial
 complexes.

 Write $\bx^{\sigma}=\prod_{i\in \sigma}x_{i}$ for some subset $\sigma\subseteq [n]$ and
 \begin{equation}
   \label{eq:3}
   \fm^{\sigma} =\langle x_{i}\, | \, i\in\sigma \rangle 
 \end{equation}
 for the monomial prime ideal corresponding to $\sigma$.
 Then the (squarefree) \textit{Alexander dual} of some monomial ideal $J = \langle
 \bx^{\sigma_{1}},\dots , \bx^{\sigma_{r}} \rangle$
 is 
 \begin{equation}
   \label{eq:4}
   J^{*}=\fm^{\sigma_{1}}\cap \dots \cap \fm^{\sigma_{r}}\, .
 \end{equation}
 It is easy to show that for any simplicial complex we have $B_{\Delta}=I_{\Delta}^*$.
 Since there is a bijection between simplicial complexes and squarefree monomial ideals
 this means that we can also have a look at the (Alexander) dual complex $\Delta^*$, which
 is determined by $I_{\Delta^*}=I_{\Delta}^*=B_{\Delta}$, i.e.~the original irrelevant
 ideal becomes the Stanley-Reisner ideal of the dual complex. It turns out that these
 considerations are indeed relevant when it comes to the computation of (reduced)
 simplicial (co)homology by Hochster's formula.

\subsection{Sheaf-Module-Correspondence and Local Coho\-mo\-logy}
\label{sec:multigraded-rings}

In order to make use of algebraic concepts, we have to reformulate the computation of
sheaf cohomology\footnote{For a short
  review of sheaf theory and sheaf cohomology have a look at the appendix of
  \cite{Cohompaper}} on the variety $X$ in terms of module theory of the Cox coordinate ring
$S$. A first step towards this is the sheaf-module correspondence, which enables us to construct quasicoherent sheaves on $X$ from any
module $M$ over $S$ that is \textit{graded} by the class group\footnote{Note that we always
  identify Picard group and class group of $X$, since in the smooth case all Weil divisors
  are already Cartier.} $\Class (X) \cong \IZ^{n-d}$. For the details of the
construction, see e.g.~\text{\S}5.3 of \cite{CoxLittleSchenk:ToricVarieties}. Since we deal with line
bundles, the only important observation is that the coordinate ring $S$ itself is $\Class
(X)$-graded, i.e.~it has a decomposition
\begin{equation}
  \label{eq:6}
  S=\bigoplus_{\alpha\in\Class (X)} S_{\alpha}\, ,\quad S_{\alpha}\cdot S_{\beta}\subset
  S_{\alpha + \beta} \,  
\end{equation}
and that the graded pieces $S_{\alpha}$ are naturally isomorphic to the sections of twisted line
bundles, namely
\begin{equation}
  \label{eq:7}
  S_{\alpha} \cong  \Gamma (X,\cO_{X}(\alpha))\, .
\end{equation}
This means that the \textit{shift} $S(\alpha)$ of $S$ defined by the grading
$S(\alpha)_{\beta}=S_{\alpha + \beta}$ gives rise to the line bundle $\cO_{X}(\alpha)$ and
therefore sheaf cohomology should also be computable from $S(\alpha)$. As it turns out,
the algebraic equivalent of sheaf cohomology is local cohomology with support on the
irrelevant ideal. The reason for this lies in the fact that the map from modules to
sheaves is not injective, since starting with $S$-modules $M$ that fulfil $\left(
  B_{\Sigma}\right)^{l}M=0$ for $l\gg 0$ leads to trivial sheaves. Taking global sections
on the sheaf side therefore in a certain way corresponds to looking at elements with
support on the irrelevant ideal on the module side.
Local cohomology is then defined completely analogous to sheaf cohomology as the
right-derived functor of the operation of taking supports.

We now introduce the necessary notions and then state the
relevant special case of Theorem 9.5.7 in \cite{CoxLittleSchenk:ToricVarieties}.
For an $S$-module $M$ and an ideal $J\subset S$ one defines the $J$-\textit{torsion
  submodule} or submodule \textit{supported on} $J$ by
\begin{equation}
  \label{eq:8}
  \Gamma_{J}(M)=\lbrace y\in M \, | \, J^{l}y =0 \, \, \text{for some } l\in \IN\,
  \rbrace\, .
\end{equation}
The $i$-th \textit{local cohomology} module of $M$ with support on $J$ is then defined to
be the module $H^{i}_{J}(M)$ obtained from any injective resolution $0 \rightarrow
I^{0}\rightarrow I^{1}\rightarrow \cdots $ of $M$ by taking the $i$-th cohomology of the
subcomplex $0\rightarrow \Gamma_{J}(I^{0})\rightarrow \Gamma_{J}(I^{1})\rightarrow \cdots
$. In particular, if $M$ is graded by $\Class (X)$, then also $\Gamma_{J} (M)$ and all
$H^{i}_{J}(M)$ will inherit this grading. The precise connection between line bundle
cohomology and local cohomology is then given by\footnote{the shift in the rank comes from
  a shift between the ordinary and the local \v{C}ech complex, see also Theorem 9.5.7 in~\cite{CoxLittleSchenk:ToricVarieties}. }
\begin{equation}
  \label{eq:9}
  H^{i}(X,\cO_{X}(\alpha)) \cong  H^{i+1}_{B_{\Sigma}}(S)_{\alpha} \quad
  \text{for } i\geq 1\, ,\alpha\in\Class (X)\, .
\end{equation}
Furthermore, there is an exact sequence
\begin{equation}
  \label{eq:10}
  0\rightarrow H^{0}_{B_{\Sigma}}(S)_{\alpha} \rightarrow S_{\alpha}
  \rightarrow H^{0}(X,\cO_{X}(\alpha)) \rightarrow
    H^{1}_{B_{\Sigma}}(S)_{\alpha}\rightarrow 0\, ,
\end{equation}
which is necessary to determine the $0$-th rank of sheaf cohomology. Because of
eq.~\eqref{eq:7} and $H^{0} (X,\cO_{X}(\alpha))=\Gamma (X,\cO_{X}(\alpha))$, the middle
map is an isomorphism. Furthermore $H^{0}_{B_{\Sigma}}(S)=0$, since $S$ has no zero
divisors and so in the special case of line bundles we get 
\begin{equation}
  \label{eq:11}
  H^{i}(X,\cO_{X}(\alpha)) \cong H^{i+1}_{B_{\Sigma}}(S)_{\alpha} \quad
  \text{for } i\geq 0\, ,\alpha\in \Class (X)\, .
\end{equation}
Before we finish this section, we want to introduce a finer grading of the Cox coordinate
ring $S$, namely the $\IZ^{n}$-grading that is given by $\text{deg}\, x_{i}=\be_{i}\in
\IZ^{n}$. The connection to the class group grading is given by the map
\begin{equation}
  \label{eq:12}
  f: \IZ^{n}\longrightarrow \Class (X) \cong \IZ^{n-d}\, ,\quad \be_{i}
  \mapsto [D_{i}] = (Q_{i}^{(1)},\dots , Q_{i}^{(n-d)})\, ,
\end{equation}
where the $Q_{i}^{(r)}$ denote the charges (or weights) of the coordinate divisor $D_{i}$
belonging to the hypersurface $\lbrace x_{i}=0 \rbrace$.
Since this finer grading is also inherited by the local cohomology modules, we may write eq.~\eqref{eq:11} as
\begin{equation}
  \label{eq:13}
  H^{i}(X,\cO_{X}(\alpha))=
  \bigoplus_{\substack{\bu\in\IZ^{n}\\f(\bu)=\alpha}}
  H^{i+1}_{B_{\Sigma}}(S)_{\bu} \,  
\end{equation}
for any $\alpha\in\Class (X)$. So the procedure would be to try and compute all $\IZ^{n}$-graded pieces of local cohomology
and at the end determine sheaf cohomology of $\cO_{X}(\alpha)$ by
summing up the contributions fulfilling $f(\bu)=\alpha$, which is a matrix equation over the
integers solvable by standard techniques using Ehrhart polynomials. In fact, this is
nothing but the rationom counting procedure of our algorithm, but we will state this later in a
more precise form.



\subsection{Simplicial Complexes and Free Resolutions}
\label{sec:simpl-compl}



The crucial link between sheaf cohomology over $X$ and the combinatorics of the associated
Stanley-Reisner ideal $I_{\Delta}$ will be a theorem by Musta\c{t}\u{a} that expresses the
dimensions of fine-graded local cohomology modules by Betti numbers of $I_{\Delta}$. But
before we state this result, we want to use this section to introduce some additional
notions from homological algebra.

Let $\Delta$ be a simplicial complex on $[n]$. For each $i\geq -1$ denote by
$F_{i}$ the set of $i$-dimensional \textit{faces} (subsets $\sigma\subseteq [n]$ of cardinality $i+1$) and let $\IC^{F_{i}}$ be the complex
vector space whose basis elements $e_{\sigma}$ correspond to all $\sigma \in F_{i}$
The \textit{reduced chain complex} of $\Delta$ is the complex
\begin{equation}
  \label{eq:5}
\wC_{\bullet}(\Delta):\quad 
  0\longleftarrow \IC^{F_{-1}}\stackrel{\ds_{0}}{\longleftarrow} \IC^{F_{0}}
  \stackrel{\ds_{1}}{\longleftarrow} \cdots \stackrel{\ds_{n-1}}{\longleftarrow} \IC^{F_{n-1}} \longleftarrow 0\, .
\end{equation}
The \textit{boundary maps} $\ds_{i}$ are defined by setting
$\sign(j,\sigma)=(-1)^{r-1}$ when $j$ is the $r$-th element of $\sigma \subseteq [n]$
written in increasing order, and
\begin{equation}
  \label{eq:14}
  \ds_{i}(e_{\sigma})=\sum_{j\in\sigma}\sign(j,\sigma)e_{\sigma\setminus j}\, .
\end{equation}
One has $\ds_{i}\circ\ds_{i+1}=0$ and therefore defines the $i$-th \textit{reduced
homology} of $\Delta$ (with coefficients in $\IC$) as 
\begin{equation}
  \label{eq:15}
  \wH_{i}(\Delta)=\kernel(\ds_{i})/\im(\ds_{i+1})\, .
\end{equation}
Note that there is a distinction between the \textit{void complex} $\lbrace \rbrace$ and
the \textit{irrelevant complex} $\lbrace \emptyset \rbrace$. One has
$\wH_{-1}(\lbrace \emptyset \rbrace )\cong \IC$, since the empty set is a face
of dimension $-1$.

Taking the vector space dual of the chain complex (and the transpose of all maps)
one gets the \textit{cochain complex} of $\Delta$ as
$\wC^{\bullet}(\Delta)=(\wC_{\bullet}(\Delta))^{*}$ with
\textit{coboundary maps} $\ds^{i}=\ds_{i}^{*}$. One similarly defines the $i$-th \textit{reduced
cohomology} of $\Delta$ as
\begin{equation}
  \label{eq:16}
  \wH^{i}(\Delta)=\kernel(\ds^{i+1})/\im(\ds^{i})\, .
\end{equation}
Since we have coefficients in $\IC$, there is a canonical isomorphism $\wH^{i}(\Delta)\cong
(\wH_{i}(\Delta))^{*}$ and thus
\begin{equation}
  \label{eq:17}
  \dim_{\IC} \wH^{i}(\Delta) = \dim_{\IC}
 \wH_{i}(\Delta)\, .
\end{equation}

We now shortly review the definitions of (minimal) free resolutions, Taylor resolution,
Betti numbers of an $S$-module $M$ and for completeness also write down (both versions of)
the Hochster formula. Details and proofs can be found in~\cite{Sturmfels}.

We begin with the definition of a free resolution. Let $V_{i}\, ,0\leq i\leq \ell$ be a
collection of free $S$-modules, i.e.~all the $V_{i}$ have the form
$\oplus_{q_{i}}S(-\ba_{q_{i}})$ with all $\ba_{q_{i}}\in \IZ^{n}$ and $S(\ba)$ denoting
the degree shift of $S$ by $\ba$. A sequence
\begin{equation}
  \label{eq:18}
  \cF_{\bullet}: \quad 0\larr V_{0}\stackrel{\phi_{1}}{\larr} V_{1}\larr \cdots
  \larr V_{\ell -1}\stackrel{\phi_{\ell}}{\larr} V_{\ell}\larr 0
\end{equation}
of free modules with maps fulfilling $\phi_{i}\circ \phi_{i+1}=0$ is called a \textit{complex}. This complex is
$\IZ^{n}$\textit{-graded} if each homomorphism is of degree $\mathbf{0}$, i.e.~for all
elements $r_{i}\in V_{i}$ one has $\text{deg}\, r_{i}=\text{deg}\, \phi_{i}(r_{i})$ in
$\IZ^{n}$.
Assuming $V_{\ell}\neq 0$, the \textit{length} of the complex $\cF_{\bullet}$ equals $\ell$.
A complex $\cF_{\bullet}$ is called a \textit{free resolution} of an $S$-module $M$ if
$\cF_{\bullet}$ is \textit{acyclic} (i.e.~exact everywhere except in homological degree $0$), where
$M=V_{0}/\im(\phi_{1})$. The Hilbert Syzygy Theorem tells us that every $S$-module has a
free resolution with length at most $n$. Since we will only look at modules $M$ of the form
$I$ or $S/I$, where $I$ is a monomial ideal of $S$, we always get free
resolutions that are naturally $\IZ^{n}$-graded. 

Define a partial order on $\IZ^{n}$ by letting $\ba \leq \bb$ if the components fulfil
$a_{i}\leq b_{i}$ for all $i\in [n]$. To state when a free resolution is minimal, we
introduce the concept of monomial matrices. These represent maps
\begin{equation}
  \label{eq:19}
  \bigoplus_{q}S(-\ba_{q})\stackrel{\phi}{\larr} \bigoplus_{p}S(-\ba_{p})
\end{equation}
between two free $S$-modules. A \textit{monomial matrix} consists of entries
$\lambda_{qp}\in \IC$ arranged in columns labeled by the \textit{source degrees} $\ba_{p}$
and rows labeled by the \textit{target degrees} $\ba_{q}$ and whose entry $\lambda_{qp}$
is zero unless $\ba_{p} \geq \ba_{q}$ in the partial order of $\IZ^{n}$. The map $\phi$
will then send the basis vector $\mathbf{1}_{\ba_{p}}$ of $S(-\ba_{p})$ to the element
$\lambda_{qp}\bx^{\ba_{p}-\ba_{q}}\cdot \mathbf{1}_{\ba_{q}}$ in $S(-\ba_{q})$. The
condition $\ba_{p}\geq \ba_{q}$ then just guarantees that $\bx^{\ba_{p}-\ba_{q}}\in S$ and
the image of $\phi$ makes sense.

Such a monomial matrix is called \textit{minimal} if $\lambda_{qp}=0$ whenever
$\ba_{p}=\ba_{q}$. Similarly, a free resolution of some module $M$ is called a
\textit{minimal free resolution} if all the maps in the resolution can be represented by
minimal monomial matrices. This means that the ranks of the free modules $V_{i}$ in a
resolution are all simultaneously minimized. In particular, any free resolution of $M$
contains the (unique up to isomorphism) minimal resolution as a subcomplex.

As an example of a free resolution, take a monomial ideal $I=\langle m_{1},\dots ,m_{t}
\rangle$ in $S$ and write $m_{\tau}=\lcm\lbrace m_{j}\, |\, j\in\tau\rbrace$ for
any $\tau\subseteq [t]$. Furthermore, set
$\ba_{\tau}=\text{deg}(m_{\tau})\in\IZ^{n}$. The \textit{full Taylor resolution} of
$I$ is based on the reduced chain complex
\begin{equation}
  \label{eq:43}
  \cT_{\bullet}(t)\ce \wC_{\bullet}(\Delta_{[t]})
\end{equation}
of the full simplex $\Delta_{[t]}$ consisting of \textit{all} subsets of $[t]$. To arrive
at the Taylor resolution, one substitutes all $\IC^{F_{j}}$ in $\cT_{\bullet}(t)$ by
$\oplus_{\tau\in F_{j}}S(-\ba_{\tau})$ and puts the boundary maps $\partial_{j}$ into a
sequence of monomial matrices $M(\partial_{j})$ with source and target labels $\ba_{\tau}$
corresponding to faces $\tau\in\Delta_{[t]}$ and entries $\lambda_{\tau,\tau\setminus
  k}=\sign (k,\tau)$ equal to the sign factors from eq.~\eqref{eq:14}. One arrives at an
acyclic complex of the form
\begin{equation}
  \label{eq:44}
  \cF^{\cT}_{\bullet}:\quad
  0 \larr S \stackrel{M(\partial_{0})}{\larr} \bigoplus_{\tau\in F_{0}}S(-\ba_{\tau})
  \stackrel{M(\partial_{1})}{\larr} \cdots \stackrel{M(\partial_{t-1})}{\larr}
  S(-\ba_{[t]})\larr 0 \, ,
\end{equation}
whose $0$-th homology equals $S/I$, so this is a free resolution of $S/I$ of length $t$.
For applications, this means that starting with some toric variety $X$, the power set of
its Stanley-Reisner ideal $I_{\Delta}$ contains all information about the full Taylor
resolution of $S/I_{\Delta}$.\footnote{Here, the term ``power set of an ideal'' stands for
  taking all possible unions of the generators. In fact, the sequences for ``remnant
  cohomology'' in the algorithm of~\cite{Cohompaper} come from the combinatorics of this
  power set and the connection with the full Taylor resolution of $S/I$ will be important
  for the proof.} Unfortunately, the Taylor resolution is almost never minimal. More
precisely, the Taylor resolution is minimal if and only if for all faces
$\sigma\in\Delta_{[r]}$ and all indices $i\in\sigma$, the monomials $m_{\sigma}$ and
$m_{\sigma\setminus i}$ are different.\footnote{For example, the Taylor resolution of the
  Stanley-Reisner ring of $X=dP_{3}$ is not minimal, since the subset $\lbrace
  m_{1},m_{2},m_{3}\rbrace = \lbrace x_{1}x_{2},x_{1}x_{3},x_{2}x_{3} \rbrace$ is among
  the generators of its Stanley-Reisner ideal, cf.~the examples in~\cite{Cohompaper}.}

All information about the minimal free resolution of some $S$-module can be encoded in
some integer numbers. If we take the complex $\cF_{\bullet}$ from \eqref{eq:18} to be a
minimal free resolution of an $S$-module $M$ and write the $V_{i}$ as
\begin{equation}
  \label{eq:20}
  V_{i}=\bigoplus_{\ba\in\IZ^{n}}S(-\ba )^{\beta_{i,\ba}}\, ,
\end{equation}
then the $i$-th \textit{Betti number} of $M$ in degree $\ba$ is the invariant
\begin{equation}
  \label{eq:21}
  \beta_{i,\ba}(M)=\beta_{i,\ba}\, .
\end{equation}
Betti numbers can also be characterized more categorically in terms of the
Tor-functor.\footnote{See~\cite{Weibel} for more details on these categorical issues.}
For two $S$-modules $M$ and $N$ one can describe the modules $\Tor_{i}^{S}(M,N)$ by
applying the functor  $\_\_\, \otimes_{S} N$ to a free resolution of $M$ and taking
homology of the resulting complex. Since all relevant notions have a generalization to the
$\IZ^{n}$-graded setting, also the Tor-modules can be given a natural $\IZ^{n}$-grading.
Intuitively speaking, the Betti numbers of $M$ then describe what survives when tensoring any
free resolution with $\IC$ and taking homology:
\begin{equation}
  \label{eq:22}
  \beta_{i,\ba}(M)=\dim_{\IC}\left(\Tor_{i}^{S}(M,\IC)_{\ba}\right)\, .
\end{equation}
Note that the tensor product over $S$ of a shifted free
module $S(-\ba )$ with the ground field $\IC\cong S/\fm $ is equal to a copy of
the ground field in degree $\ba\in\IZ^{n}$:
\begin{equation}
  \label{eq:45}
  S(-\ba )\otimes_{S}\IC \cong \IC (-\ba )
\end{equation}
This means that one has an easy description of the degree $\ba$ piece of a tensored
resolution $\cF_{\bullet}\otimes_{S}\IC$. All copies of $S(-\ba )$ for some
$\ba\in\IZ^{n}$ that were present in the resolution become one-dimensional vector spaces
$\IC (-\ba)$ and since all maps between source and target degrees with $\ba_{p}\neq
\ba_{q}$ become zero, one can restrict to degree $\ba$ by just looking at the subcomplex
of $\cF_{\bullet}\otimes_{S}\IC$ made up of the spaces $\IC(-\ba)$. These considerations
will play a role in the proof of our theorem later.

The Betti numbers of a monomial ideal $I_{\Delta}\subseteq S$ may be calculated in
different ways. One possibility is to take the (co)homology of certain simplicial
subcomplexes of the associated complex $\Delta$ (resp.~the Alexander dual $\Delta^{*}$).
This is described by the Hochster formula (resp.~its dual version). 

For each $\sigma \subseteq [n]$ define the \textit{restriction} of a simplicial complex
$\Delta$ to $\sigma$ by
\begin{equation}
  \label{eq:23}
  \Delta |_{\sigma} = \lbrace \tau \in\Delta \, |\, \tau\subseteq \sigma \rbrace \, .
\end{equation}
For $\sigma\subseteq [n]$ write $\widetilde{\sigma}\in\IZ^{n}$ for the (squarefree) degree
with components $\widetilde{\sigma}_{i}=1$ if $i\in\sigma$ and $\widetilde{\sigma}_{i}=0$
otherwise. Since the meaning can always be inferred from the context, we subsequently omit
the tilde and write $\sigma$ also for the element in $\IZ^{n}$. Treating
$I_{\Delta}$ and $S/I_{\Delta}$ as (graded) $S$-modules, their Betti numbers lie only in
squarefree degrees $\sigma$ and can be calculated by the \textit{Hochster formula}
\begin{equation}
  \label{eq:24}
  \beta_{i-1,\wsig}(I_{\Delta})=\beta_{i,\wsig}(S/I_{\Delta})=\dim_{\IC}\,\wH^{\lvert
  \sigma \rvert -i-1}(\Delta |_{\sigma})\, .
\end{equation}
There is also a description of these Betti numbers in terms of the Alexander dual
complex $\Delta^{*}$. For any $\sigma \subseteq [n]$ define the \textit{link} of $\sigma$
inside the simplicial complex $\Delta$ to be 
\begin{equation}
  \label{eq:25}
  \link_{\Delta}(\sigma)=\lbrace \tau \in \Delta \, |\, \tau\cup\sigma \in\Delta
  \,\,\text{and}\,\, \tau\cap\sigma=\emptyset \rbrace\, .
\end{equation}
Furthermore, denote the \textit{complement} of $\sigma\subseteq [n]$ by
$\overline{\sigma}=[n]\setminus\sigma$. Then the \textit{dual Hochster formula} states that
\begin{equation}
  \label{eq:26}
  \beta_{i,\wsig}(I_{\Delta})=\beta_{i+1,\wsig}(S/I_{\Delta})=\dim_{\IC}\,
  \wH_{i-1}(\link_{\Delta^{*}}(\overline{\sigma}))\, .
\end{equation}
Because of eq.~\eqref{eq:17}, in all of these formulas simplicial homology may be treated
for cohomology when computing Betti numbers.

\subsection{Local Cohomology and Betti Numbers}
\label{sec:local-cohom-betti}

Now we come to the connection between local cohomology and Betti numbers stated
in~\cite{EMS-Paper}. Let $\ba\in\IZ^{n}$ and define the set of indices with negative
entries by $\negind(\ba)=\lbrace i\in [n]\, |\, \ba_{i}<0 \rbrace\subseteq [n]$. As it turns out, the
graded parts of local cohomology of $S$ with support on the irrelevant ideal of some toric
variety do only depend on the negative entries in the degree, i.e.~for
$\ba,\bb\in\IZ^{n}$ with $\negind(\ba )=\negind(\bb )$, one has
\begin{equation}
  \label{eq:27}
  H^{i}_{B_{\Sigma}}(S)_{\ba}\cong H^{i}_{B_{\Sigma}}(S)_{\bb}\, 
\end{equation}
or in the notation of the last section
\begin{equation}
  \label{eq:30}
  H^{i}_{B_{\Sigma}}(S)_{\ba}\cong H^{i}_{B_{\Sigma}}(S)_{-\wsig}\, 
\end{equation}
when $\negind (\ba )=\sigma$. Therefore, one only has to compute
$H^{i}_{B_{\Sigma}}(S)_{-\wsig}$ for all $\sigma\subseteq [n]$. The decisive result is now
that via Alexander duality there is a direct relation between these graded pieces of local
cohomology and the Betti numbers of $S/I_{\Sigma}$, since (cf.~Corollary 1.2
in~\cite{EMS-Paper})
\begin{equation}
  \label{eq:28}
  H^{i}_{B_{\Sigma}}(S)_{-\wsig}\cong \Tor_{|\sigma | - i +1}^{S}(S/I_{\Sigma},\IC )_{\wsig} \, ,
\end{equation}
and application of eq.~\eqref{eq:22} then gives
\begin{equation}
  \label{eq:29}
  \dim_{\IC}\left( H^{i}_{B_{\Sigma}}(S)_{-\wsig}\right) = \beta_{|\sigma | -i +
    1,\wsig}(S/I_{\Sigma})\, .
\end{equation}
Inserting this into eq.~\eqref{eq:13}, we finally get a closed formula for the line bundle
cohomology $h^{i}(\alpha)\ce \dim_{\IC}H^{i}(X,\cO_{X}(\alpha))$ in terms of ``chamber
factors'' and Betti numbers as
\begin{equation}
  \label{eq:31}
  h^{i}(\alpha)
  =\sum_{\substack{ \bu\in\IZ \\ f(\bu )=\alpha }}\dim_{\IC}\,
  H^{i+1}_{B_{\Sigma}}(S)_{-\negind(\bu)} 
  =\sum_{\sigma \subseteq [n]}\vert (\alpha ,\sigma ) \vert \cdot \beta_{|\sigma |
    -i,\wsig}(S/I_{\Sigma})\, ,
\end{equation}
where $| (\alpha ,\sigma ) |$ counts the number of elements in the set
\begin{equation}
  \label{eq:32}
  (\alpha ,\sigma ) =\lbrace \bu\in\IZ^{n}\, |\, f(\bu )=\alpha\, ,\, \negind(\bu
  )=\sigma \rbrace \, .
\end{equation}
This is the starting point for the proof of our algorithm to which we turn next.

\section{Proof of the Conjecture}

We first want to restate the conjecture from~\cite{Cohompaper} in a more precise form.
Therefore, we divide it into two parts, where the first part refers to the possible
restriction to degrees in the power set of the Stanley-Reisner ideal by a vanishing of
Betti numbers while the second gives a mathematically precise way to compute the ``remnant
cohomology''.

Let $X$ be a complete simplicial smooth normal toric variety in the sense
of~\cite{CoxLittleSchenk:ToricVarieties} endowed with all necessary data listed in
Section~\ref{sec:norm-toric-vari}.
In particular, let the Stanley-Reisner ideal $I_{\Sigma}=\langle \cS_{1}, \dots
,\cS_{t}\rangle$ be generated by $t$ different squarefree monomials $\cS_{i}$ in the
coordinate ring $S=\IC [\bx ]=\IC [x_{1},\dots , x_{n}]$. For all $\tau\subseteq [t]$ set
$\cS_{\tau}=\lcm\lbrace \cS_{i}\, |\, i\in \tau \rbrace$ and denote by $\ba_{\tau}=\degree
(S_{\tau})\in\IZ^{n}$ its degree in the natural $\IZ^{n}$-grading of $S$. Furthermore, the
\textit{degrees in the power set of} $I_{\Sigma}$ are given by
\begin{equation}
  \label{eq:33}
  \cP (I_{\Sigma})=\lbrace \ba_{\tau}\, | \,\tau\in [t] \rbrace\, .
\end{equation}

Recall from Section \ref{sec:simpl-compl} that the full Taylor resolution
$\cF^{\cT}_{\bullet}$ of $S/I_{\Sigma}$ is based on the Taylor complex $\cT_{\bullet}(t)$
which is just the reduced chain complex of the full simplex $\Delta_{[t]}$. For some
$\sigma\subseteq [n]$ define the (relative) simplicial subcomplex
\begin{equation}
  \label{eq:41}
  \Gamma^{\sigma} =\lbrace \tau\in [t]\, | \, \ba_{\tau}=\wsig \rbrace\, .
\end{equation}
of the full simplex $\Delta_{[t]}$. One can define maps between the sets $F_{j}(\Gamma^{\sigma})$ of
$j$-dimensional faces of this subcomplex similar to eq.~\eqref{eq:14} by
\begin{equation}
  \label{eq:42}
  \phi_{j}: F_{j}(\Gamma^{\sigma})\longrightarrow F_{j-1}(\Gamma^{\sigma})\, ,\quad e_{\tau}\mapsto \sum_{k\in\tau}\sign
  (k,\tau)e_{\tau\setminus k}\, ,
\end{equation}
where $e_{\tau\setminus k}=0$ if $\tau\setminus k\notin\Gamma^{\sigma}$ and $\sign
(k,\tau )=(-1)^{s-1}$ when $k$ is the $s$-th element of $\tau\subseteq [t]$ written in
increasing order. Since this is just the restriction of the boundary maps in
$\cT_{\bullet}(t)$, it is easy to see that this yields a well-defined complex
$\wC_{\bullet}(\Gamma^{\sigma})$ with associated (reduced relative) homology
$\wH_{\bullet}(\Gamma^{\sigma})$.

\begin{theorem}
  Let $\alpha\in\Class(X)$ and $h^{i}(\alpha )\ce \dim_{\IC}\,
  H^{i}(X,\cO_{X}(\alpha ))$. Then we have:
  \begin{enumerate}[(a)]
  \item For all $\sigma\subseteq [n]$ with $\wsig\notin\cP (I_{\Sigma})$ the associated
    Betti numbers vanish, i.e.~
    \begin{equation}
      \label{eq:39}
      \beta_{r,\wsig }(S/I_{\Sigma})=0 \quad \text{for all } r\geq 0\, .
    \end{equation}
    Therefore, one may restrict the sum in eq.~\eqref{eq:31} to
    \begin{equation} 
      \label{eq:37}
      h^{i}(\alpha )=\sum_{\substack{\sigma \subseteq [n]\\ \wsig\in\cP
          (I_{\Sigma})}} | (\alpha ,\sigma ) |\, \cdot \beta_{| \sigma | -i
        ,\wsig}(S/I_{\Sigma})\, ,
    \end{equation}
    where $| (\alpha ,\sigma ) |$ counts the number of elements in the set\footnote{We want to note that the
    set $(\alpha ,\sigma)$ corresponds to all so-called \textit{rationoms} $\bx^{\bu}$
    with $f(\bu )=\alpha$ and precisely the coordinates $x_{i}$ with $i\in\sigma$ standing
    in the denominator. Intuitively, these rationoms can be interpreted as
    ``representatives'' of \v{C}ech cohomology on intersections of open sets in the toric
    variety, cf.~Section 2.2 of~\cite{Cohompaper}.}
    \begin{equation}
      \label{eq:34}
      (\alpha ,\sigma ) =\lbrace \bu\in\IZ^{n}\, |\, f(\bu )=\alpha\, ,\, \negind(\bu
      )=\sigma \rbrace \, .
    \end{equation}
  \item The Betti numbers $\beta_{r , \wsig}(S/I_{\Sigma})$ can be calculated from the
    degree $\wsig$ part of the full Taylor resolution. This can be described in terms of
    the (relative) subcomplex $\Gamma^{\sigma}$ as\footnote{This corresponds to the
      sequences for ``remnant cohomology'' in~\cite{Cohompaper}. By counting the number of
    times that a fixed denominator $\bx^{\sigma}$ appears in rank $r$ of the
    Stanley-Reisner power set, one gets the number of $(r-1)$-faces of the complex
    $\Gamma^{\sigma}$. If one also takes notice of the different combinations of
    Stanley-Reisner generators that lead to this denominator, one can write down the maps
    in~\eqref{eq:42} and gets a well-defined complex.}
    \begin{equation}
      \label{eq:40}
      \beta_{r, \wsig}(S/I_{\Sigma})= \dim_{\IC}\, \wH_{r-1} 
      (\Gamma^{\sigma} ) \, .
    \end{equation}
  \end{enumerate} 
\end{theorem}
\begin{proof}[Proof of the Theorem:]
  To get the desired Betti numbers, we can tensor the Taylor resolution of $S/I_{\Sigma}$
  with $\IC \cong S/\fm$, extract the degree $\wsig$ part and take homology,
  cf.~\eqref{eq:22}. As we have described around eq.~\eqref{eq:45}, the tensored
  resolution will just be made up of vector spaces $\IC (-\ba )$ of degree $\ba$ at the
  locations of $S(-\ba)$ in the original resolution. Considering the maps of the tensored
  resolution, note that all entries $\lambda_{\tau, \rho}$ with $\ba_{\tau}\neq
  \ba_{\rho}$ become zero, since they correspond to multiplication by
  $\bx^{\ba_{\rho}-\ba_{\tau}}=0\,$ in $S/\fm$. So we can easily extract graded parts. In
  particular, since we started with a Taylor complex, the restriction of the tensored
  resolution to its degree $\sigma$ part will consist of all occurrences of
  $\IC(-\ba_{\tau})$ with $\ba_{\tau}=\sigma$ and therefore be equivalent to the
  (relative) complex $\Gamma^{\sigma}$ with maps as in eq.~\eqref{eq:42}, i.e.
  \begin{equation}
    \label{eq:54}
    \left( \cF^{\cT}_{\bullet}\otimes_{S}\IC \right)_{\sigma} \cong \wC_{\bullet-1}(\Gamma^{\sigma})\, .
  \end{equation}
  The shift by one in the homological degree on the right hand side comes from the fact
  that the empty set in the Taylor complex lies in homological degree $-1$ while the corresponding
  free module $S$ in the full Taylor resolution of eq.~\eqref{eq:44} lies in homological
  degree $0$. To get $\Tor^{S}_{r}(S/I_{\Sigma},\IC)_{\wsig}$,
  we still have to take homology, yielding
  \begin{equation}
    \label{eq:47}
    \Tor^{S}_{r}(S/I_{\Sigma},\IC)_{\wsig} \cong \wH_{r-1} (\Gamma^{\sigma})\, ,
  \end{equation}
  which finally implies eq.~\eqref{eq:40} by taking dimensions. If $\wsig\notin\cP
  (I_{\Sigma})$, the complex $\Gamma^{\sigma}$ is void and therefore has zero homology.
  This implies that the respective Betti numbers vanish and the expression for
  $h^{i}(\alpha)$ then follows from eq.~\eqref{eq:31}.
\end{proof}

We also want to point out the simple connection of our theorem to the statement
in~\cite{Jow}. To see this, fix some $\sigma\subseteq [n]$ and then arrange the generators
of $I_{\Sigma}=\langle \cS_{1},\dots ,\cS_{m},\cS_{m+1},\dots ,\cS_{t}\rangle $ such that
$\cS_{i}$ divides $\bx^{\wsig}$ for all $i\leq m$ and $m$ is maximal with this property.
Now define a subcomplex of $\Delta_{[m]}$ by
\begin{equation}
  \label{eq:38}
  \Lambda^{\sigma}= \lbrace \mu\in [m]\, |\, \ba_{\mu} < \wsig \rbrace\, ,
\end{equation}
where ``$<$'' is the partial order of $\IZ^{n}$. It is easy to see that
$\Gamma^{\sigma}$ also lies in $\Delta_{[m]}$ and actually is the relative complex
\begin{equation}
  \label{eq:48}
  \Gamma^{\sigma} = \Delta_{[m]} / \Lambda^{\sigma}\, .
\end{equation}
Equivalence of our theorem to the theorem of Jow is then provided by the 
\begin{lemma}
  The homologies of the reduced chain complexes of $\Lambda^{\sigma}$ and
  $\Gamma^{\sigma}$ are isomorphic up to a shift in the homological degree, i.e.
  \begin{equation}
    \label{eq:49}
    \wH_{i}(\Gamma^{\sigma}) \cong \wH_{i-1} (\Lambda^{\sigma}) \, .
  \end{equation}
\end{lemma}
\begin{proof}
  Because of eq.~\eqref{eq:48}, one has an exact sequence of reduced chain complexes
  \begin{equation}
    \label{eq:50}
    0 \rarr \wC_{\bullet}(\Lambda^{\sigma})\rarr \wC_{\bullet}(\Delta_{[m]}) \rarr
    \wC_{\bullet}(\Gamma^{\sigma}) \rarr 0\, ,
  \end{equation}
  where the first map is the natural inclusion. Since $\Delta_{[m]}$ is a full simplex,
  its homology vanishes and the long exact sequence in homology then yields 
  \begin{equation}
    \label{eq:51}
    \wH_{\bullet}(\Gamma^{\sigma})\cong \wH_{\bullet - 1} (\Lambda^{\sigma})\, .
  \end{equation}
\end{proof}

\section{Conclusions and Open Questions}
\label{sec:altern-algor}

As one can already tell from the Hochster formulas, there are many different complexes
from which the Betti numbers of a squarefree monomial ideal can be calculated. All of these have
their computational advantages and disadvantages. We want to mention the version of Smith
in~\cite{Smith}, Proposition 3.2, which is implemented in the Macaulay2 package
``NormalToricVarieties.m2''. It follows from an application of the Hochster
formula~\eqref{eq:24} to the results of Section \ref{sec:local-cohom-betti}. Since this
version of the Hochster formula contains a restriction of the complex $\Delta$, it also
requires the knowledge of the maximal faces of $\Delta$ and therefore of the irrelevant
ideal of the toric variety. In contrast to this, a nice feature of the new algorithm is
that it is only based on the combinatorics of the Stanley-Reisner ideal, which also
explains why the vanishing of Betti numbers with degrees not in the power set of
$I_{\Sigma}$ was not detected before.

Since we aim for an efficient computation of sheaf cohomology, it would be desirable to
get rid of explicit maps in the complexes that are needed for the computation of Betti
numbers. The perfect statement would be that for a fixed degree there can at most be one
nontrivial Betti number, as then the alternating sum of dimensions and some considerations
about the possible homological degree of the contribution would suffice to determine it.
This statement is strikingly similar to the Eagon-Reiner-Theorem for Cohen-Macaulay
Stanley-Reisner rings, see Theorem 5.56 in~\cite{Sturmfels}. But for a normal toric
variety, the fact that $S/I_{\Sigma}$ is Cohen-Macaulay only implies that the Alexander
dual $B_{\Sigma}$ has the desired property. The question is now, if some of these
simplifications carry over to the resolution of $I_{\Sigma}$. There are some theorems
concerning such a duality for resolutions, but the only precise statements can be made for
so-called extremal Betti numbers, see~\cite{Charalambous}. It is not clear to the authors
how many non-extremal Betti numbers are around in the general case that spoil the broth. 

In fact, by running our program for all ``boundary divisors'' with charge
\begin{equation}
  \label{eq:55}
\alpha_{\sigma}=f(-\sigma)\in\Cl (X),\quad \sigma\subseteq [n]  
\end{equation}
for certain (perhaps too special) examples\footnote{For $X=dP_{3}$ the
  divisor $D=-3H-X-Y-Z$ that we took as an example in~\cite{Cohompaper} is such a boundary
  divisor with the same charge as $x_{1}^{-1}x_{2}^{-1}x_{3}^{-1}$.} of toric varieties, we learned that they
always had at most one non-trivial cohomology and the number coincided with the only
non-zero Betti number in degree $\sigma$. This means that for these examples we not only
found a certain uniqueness of the cohomological degree that a ``denominator'' can
contribute to, but also that either there is only one non-empty ``chamber'', namely $| (\alpha
,\sigma )|=1$, or all other solutions 
\begin{equation}
  \label{eq:56}
  \lbrace \bu\in\IZ^{n} \, | \, f(\bu )=\alpha_{\sigma}\, , \bu \neq -\sigma \rbrace
\end{equation}
lie in chambers with vanishing Betti numbers. This implies that in these cases one has
some kind of ``Serre duality for Betti numbers'', meaning that when starting with Serre
duality for line bundles
\begin{equation}
  \label{eq:52}
  h^{i}(\alpha_{\sigma})=h^{d-i}(\alpha_{\overline{\sigma}})\, ,
\end{equation}
where $\overline{\sigma}$ denotes the complement of $\sigma$, one also gets
\begin{equation}
  \label{eq:53}
  \beta_{|\sigma| - i,\sigma}(S/I_{\Sigma})=\beta_{n-d-(|\sigma|-i),\overline{\sigma}}(S/I_{\Sigma})\, .
\end{equation}
In particular, by an application of our theorem, there is another vanishing of Betti
numbers
\begin{equation}
  \label{eq:57}
  \beta_{r,\sigma}(S/I_{\Sigma})=0\, \,\, \text{whenever }\,\overline{\sigma}\notin\cP (I_{\Sigma})\, 
\end{equation}
 in these examples.

 We would very much appreciate comments or explicit (counter-)examples concerning these
 issues and hope that the computational power of the algorithm will lead to some new
 insights in the diverse areas of mathematics and physics it can be applied to.

\subsection*{Acknowledgement}

We owe a lot to our supervisor Ralph Blumenhagen, who formulated the conjecture merely out
of his clear-sighted intuition. It is also a pleasure to thank Benjamin Jurke for helpful
discussions, proof-reading and, most importantly, for the quick and robust implementation
of the algorithm. Furthermore, we would like to thank Hal Schenck for supplying us with
hints and Manfred Herbst for interesting discussions about the possibilities to compute
sheaf cohomology in terms of empty D-branes in derived categories.


\clearpage
\providecommand{\href}[2]{#2}\begingroup\raggedright\endgroup



\begin{thebibliography}{10}

\bibitem[BCP]{Charalambous}
D.~Bayer, H.~Charalambous, S.~Popescu, ``{Extremal Betti Numbers and Applications to
  Monomial Ideals },'' \href{http://arxiv.org/abs/math/9804052}{\tt arXiv:math/9804052}.

\bibitem[BJRR1]{Cohompaper}
R.~Blumenhagen, B.~Jurke, T.~Rahn, H.~Roschy, ``{Cohomology of line bundles: A
  Computational Algorithm},'' \href{http://arxiv.org/abs/1003.5217}{{\tt arXiv:hep-th/1003.5217 }}.

\bibitem[BJRR2]{PaperSecondPart}
R.~Blumenhagen, B.~Jurke, T.~Rahn, and H.~Roschy, ``{Cohomology of Line
  Bundles: Applications}.'' work in progress.

\bibitem[CGH]{Cvetic:2010rq}
M.~Cvetic, I.~Garcia-Etxebarria, and J.~Halverson, ``{Global F-theory Models:
  Instantons and Gauge Dynamics},''
\href{http://arxiv.org/abs/1003.5337}{{\tt arXiv:1003.5337 [hep-th]}}.

\bibitem[CLS]{CoxLittleSchenk:ToricVarieties}
D.~A. Cox, J.~B. Little, and H.~Schenck, ``{Toric Varieties}.'' unpublished
  book-in-progress, available at {\tt\url{http://www.cs.amherst.edu/~dac/toric.html}}.

\bibitem[EMS]{EMS-Paper}
D.~Eisenbud, M.~Musta\c{t}\u{a}, M.~Stillman, ``{Local Cohomology With Monomial
  Support},''{{\em J. Symbolic Comput.}  {\bf  29} (2000)  no. 4-5, 583-600},
\href{http://arxiv.org/abs/math/0001159}{\tt arXiv:math/0001159}.

\bibitem[GS]{Macaulay2}
D.~Grayson and M.~Stillman, ``{Macaulay2, a software system for research in algebraic
  geometry},'' available by ftp at \href{http://www.math.uiuc.edu/Macaulay2/}{http://www.math.uiuc.edu/Macaulay2/}.

\bibitem[Jow]{Jow}
S.-Y.~Jow, ``{Cohomology of Toric Line Bundles via Simplicial Alexander Duality},''
\href{http://arxiv.org/abs/1006.0780}{\tt arXiv:1006.0780 [math.AG]}.

\bibitem[MaS]{Smith}
D.~MacLagan and G.~Smith, ``{Multigraded Castelnuovo-Mumford Regularity},''
\href{http://dx.doi.org/10.1515/crll.2004.040}{{\em J. reine angew. Math.} {\bf  571} (2004)  179-212},
\href{http://arxiv.org/abs/math/0305214}{\tt arXiv:math/0305214}.


\bibitem[MiS]{Sturmfels}
E.~Miller and B.~Sturmfels, ``{Combinatorial Commutative Algebra},'' Springer, 2005.

\bibitem[Wei]{Weibel}
C.~Weibel, ``{An introduction to homological algebra},'' Cambridge University Press, 1994.





















\end{thebibliography}
\end{document}